\documentclass{sig-alternate}
\usepackage[paper=letterpaper,width=7in,height=9.25in,centering]{geometry}

\usepackage[ruled,vlined]{algorithm2e}
\usepackage{tikz}
\usepackage{url}
\usepackage{lhelp}

\def\C {\ensuremath{\mathbb{C}}}

\def\K {\ensuremath{\mathbf{k}}}

\def\Q {\ensuremath{\mathbb{Q}}}
\def\R {\ensuremath{\mathbb{R}}}
\def\T {\ensuremath{\mathfrak{T}}}
\def\S {\ensuremath{\mathfrak{S}}}
\newcommand{\NN}{\mbox{$N_{\geq}$}}
\newcommand{\PP}{\mbox{$P_{>}$}}
\newcommand{\HH}{\mbox{$H_{\neq}$}}
\def\QQ {\ensuremath{\mathcal{Q}}}

\newcommand{\discrim}[1]{\mbox{{\rm discrim}$(#1)$}}

\newcommand{\init}[1]{\mbox{{\rm init}$(#1)$}}

\newcommand{\mdeg}[1]{\mbox{{\rm mdeg}$(#1)$}}
\newcommand{\mvar}[1]{\mbox{{\rm mvar}$(#1)$}}

\newcommand{\sat}[1]{\mbox{{\rm sat}$(#1)$}}

\newcommand{\oproj}[1]{\mbox{{\rm oproj}$(#1)$}}
\newcommand{\der}[1]{\mbox{{\rm der}$(#1)$}}

\newcommand{\Regularize}[1]{\mbox{{\sf Regularize}$(#1)$}}
\newcommand{\Triangularize}[1]{\mbox{{\sf Triangularize}$(#1)$}}
\newcommand{\RealTriangularize}[1]{\mbox{{\sf RealTriangularize}$(#1)$}}
\newcommand{\BorderPolynomial}[1]{\mbox{{\sf BorderPolynomialSet}$(#1)$}}
\newcommand{\BP}[1]{\mbox{{\sf BorderPolynomial}$(#1)$}}
\newcommand{\GenerateRegularSas}[1]{\mbox{{\sf GenerateRegularSas}$(#1)$}}
\newcommand{\GeneratePreRegularSas}[1]{\mbox{{\sf GeneratePreRegularSas}$(#1)$}}

\newcommand{\LazyRealTriangularize}[1]{\mbox{{\sf LazyRealTriangularize}$(#1)$}}
\newcommand{\GenerateFormula}[1]{\mbox{{\sf GenerateFormula}$(#1)$}}
\newcommand{\RealRootCounting}[1]{\mbox{{\sf RealRootCounting}$(#1)$}}
\newcommand{\ReviseFormula}[1]{\mbox{{\sf Disjunction}$(#1)$}}
\newcommand{\SamplePoints}[1]{\mbox{{\sf SamplePoints}$(#1)$}}
\newcommand{\factor}[1]{\mbox{{\rm factor}$(#1)$}}

\newcommand{\RegularChains}{{\tt RegularChains}}

\def\KK {\ensuremath{\mathbf{K}}}
\def\u {\ensuremath{\mathbf{u}}}
\def\x {\ensuremath{\mathbf{x}}}
\def\y {\ensuremath{\mathbf{y}}}

\newtheorem{Theorem}{Theorem}
\newtheorem{Lemma}{Lemma}

\newtheorem{Proposition}{Proposition}

\newcommand{\mc}[1]{\mathcal{#1}}
\newcommand{\mr}[1]{\mathrm{#1}}
\newcommand{\bb}[1]{\mathbb{#1}}

\def\R {\ensuremath{\mathbb{R}}}
\def\Q {\ensuremath{\mathbb{Q}}}
\def\C {\ensuremath{\mathbb{C}}}

\DeclareMathOperator{\res}{res}

\DeclareMathOperator{\oaf}{oaf}

\newif\ifcomment
\commentfalse
\newif\ifdraft
\draftfalse

\begin{document}

\title{Triangular Decomposition of Semi-algebraic Systems}
\numberofauthors{6} 
\author{
\alignauthor
Changbo Chen\\
       \affaddr{University of Western Ontario}\\
       \email{cchen252@csd.uwo.ca}
\alignauthor
James H. Davenport\\
       \affaddr{University of Bath}\\
       \email{J.H.Davenport@bath.ac.uk}
\alignauthor 
John P. May\\
       \affaddr{Maplesoft}\\
       \email{jmay@maplesoft.com}
\and
\alignauthor 
Marc Moreno Maza\\
       \affaddr{University of Western Ontario}\\
       \email{moreno@csd.uwo.ca}
\alignauthor
Bican Xia\\
       \affaddr{Peking University}\\
       \email{xbc@math.pku.edu.cn}
\alignauthor
Rong Xiao\\
       \affaddr{University of Western Ontario}\\
       \email{rong@csd.uwo.ca}
}


\conferenceinfo{ISSAC 2010,}{25--28 July 2010, Munich, Germany.}
\CopyrightYear{2010}
\crdata{978-1-4503-0150-3/10/0007}

\maketitle
\begin{abstract}
Regular chains and triangular decompositions are
fundamental and well-developed tools for describing the complex solutions 
of polynomial systems. 
This paper proposes adaptations of these tools focusing
on solutions of the real analogue:
semi-algebraic systems.

We show that any such system can be decomposed into finitely many 
{\em regular semi-algebraic systems}.
We propose two specifications of such a decomposition
and present corresponding algorithms.
Under some assumptions, one type of decomposition can be computed in singly 
exponential time w.r.t.\ the number of variables.  
We implement our algorithms and 
the experimental results illustrate their effectiveness.
\end{abstract}


\section{Introduction}
Regular chains, the output of triangular decompositions of systems 
of polynomial equations, enjoy remarkable properties.
Size estimates play in their favor~\cite{DaSKac09}
and permit the design of modular~\cite{DMSWX05a} and
fast~\cite{LiMorenoPan09} methods for computing
triangular decompositions.
These features stimulate the development of algorithms
and software for solving polynomial systems
via triangular decompositions.

For the fundamental case of 
semi-algebraic systems
with rational number coefficients, to which this paper is devoted,
we observe that several algorithms for studying the real solutions
of such systems take advantage of the structure of a regular chain.
Some are specialized to isolating the real solutions of
systems with finitely many complex 
solutions~\cite{xz06,CGY07,BoulierChenLemaireMorenoMaza09}.
Other algorithms deal with parametric
polynomial systems via real root classification ({\small RRC})~\cite{yhx01}
or with arbitrary systems via cylindrical
algebraic decompositions ({\small CAD})~\cite{CMXY09}.

In this paper,
we introduce the notion of a {\em regular semi-algebraic system}, 
which in broad terms  is the ``real''
counterpart of the notion of a regular chain.
Then we define two notions of a 
{\em decomposition of a semi-algebraic system}:
one that we call {\em lazy triangular decomposition}, 
where the ana\-lysis of components of strictly smaller dimension 
is deferred, and one that we call 
 {\em full triangular decomposition} 
where all cases are worked out.
These decompositions are 
obtained by combining triangular decompositions
of algebraic sets over the complex field with a special 
Quantifier Elimination ({\small QE}) method
based on {\small RRC} techniques.

\smallskip\noindent{\small \bf Regular semi-algebraic system.}
Let $T$ be a regular chain of $\Q[x_1, \ldots, x_n]$ for some
ordering of the variables $\x = x_1, \ldots, x_n$.
Let $\u = u_1, \ldots, u_d$ and $\y = y_1, \ldots, y_{n-d}$ 
designate respectively the
variables of $\x$ that are free and algebraic w.r.t.\ $T$.
Let $P \subset \Q[\x]$ be finite such that
each polynomial in $P$ is regular w.r.t.\ the saturated ideal of $T$.
Define $\PP :=\{p>0\mid p\in P\}$.
Let $\QQ$ be a quantifier-free formula of $\Q[\x]$
involving only the variables of $\u$.
We say that $R := [\QQ, T, \PP]$ is a {\em regular semi-algebraic system} if:
\begin{itemizeshort}
\item[$(i)$] $\QQ$ defines a non-empty open 
             semi-algebraic set $S$ in $\R^d$,
\item[$(ii)$] the regular system $[T, P]$ specializes well at every point $u$ of $S$
    (see Section~\ref{sect:preliminary} for this notion),
\item[$(iii)$] at each point $u$ of $S$, 
the specialized system $[T(u), P(u)_{>}]$ 
has at least one real zero.
\end{itemizeshort}
The zero set of $R$, denoted by ${Z}_{\R}(R)$,
is defined as
the set of points $(u, y) \in {\R}^d \times {\R}^{n-d}$ such that 
$\QQ(u)$ is true and $t(u, y)=0$, $p(u, y)>0$, for all $t\in T$
and all $p\in P$.

\smallskip\noindent{\small \bf Triangular decomposition 
of a semi-algebraic system.}
In Section~\ref{sect:specification} we show
that the zero set of any semi-algebraic system $\S$ can be decomposed 
as a finite union (possibly empty) 
of zero sets of regular semi-algebraic systems.
We call such a decomposition a {\em full triangular decomposition} 
(or simply {\em triangular decomposition} when clear from context)
of  $\S$, 
and denote by {\sf RealTriangularize} an algorithm to compute it.
The proof of our statement relies on triangular decompositions
in the sense of Lazard 
(see Section~\ref{sect:preliminary} for this notion)
for which it is not known whether or not they  
can be computed in singly exponential time w.r.t.\ the number of variables.
Meanwhile, we are hoping to obtain an algorithm
for decomposing semi-algebraic systems 
(certainly under some genericity assumptions)
that would fit in that complexity class.
Moreover, we observe that, in practice, 
full triangular decompositions
are not always necessary and that 
providing information about the components of maximum
dimension is often sufficient.
These theoretical and practical motivations 
lead us to a 
weaker notion of a decomposition
of a semi-algebraic system.

\newcommand{\BB}{\mbox{$B_{\neq}$}}

\smallskip\noindent{\small \bf Lazy triangular decomposition 
of a semi-algebraic system.}
Let $\S=[F,\NN,\PP,\HH]$ (see Section~\ref{sect:specification}
for this notation) be a semi-algebraic system
of  ${\Q}[\x]$ and $Z_{\R}(\S) \subseteq {\R}^n$
be  its zero set.
Denote by $d$ the dimension of the constructible  set
$
\{x\in\C^n\mid f(x)=0,g(x)\neq0,\mbox{~for all~} f\in F, g\in P\cup H\}.
$ 
A finite set of regular semi-algebraic systems $R_i$, $i=1\cdots t$
is called a \emph{lazy triangular decomposition} 
of $\S$ if
\begin{itemizeshort}
\item  $\cup_{i=1}^t Z_{\R}(R_i)\subseteq Z_{\R}(\S)$ holds, and
\item there exists $G \subset {\Q}[\x]$ such that 
     the real-zero set $Z_{\R}(G) \allowbreak\subset {\R}^n$
     contains 
$Z_{\R}(\S) \setminus \left(\cup_{i=1}^{t}Z_{\R}({R}_i)\right)$
    and the complex-zero set $V(G) \subset \bb{C}^n$
    either is empty or
   has dimension less than $d$.
\end{itemizeshort}
We denote by {\sf LazyRealTriangularize} an algorithm 
computing such a decomposition. 
In the implementation presented hereafter,
{\sf LazyRealTriangularize} 
outputs additional information in order to 
continue the computations and obtain
a full triangular decomposition, if needed.
This additional information appears in the form
of {\em unevaluated function calls}, explaining
the usage of the adjective {\em lazy}
in this type of decompositions.

\smallskip\noindent{\small {\bf Complexity results for lazy 
triangular decomposition.}}
In Section~\ref{sec:BPRTD}, we provide a running time 
estimate for computing a lazy triangular decomposition
of the semi-algebraic system $\S$ 
when $\S$  has no inequations nor inequalities,
(that is, when $\NN = \PP = \HH = \emptyset$ holds)
and when $F$ generates a strongly equidimensional ideal of dimension $d$.
We show that one can compute such a decomposition
in time singly exponential w.r.t.\ $n$.
Our estimates are not sharp and are just meant to reach
a singly exponential bound.
We rely on the work of J. Renagar~\cite{Ren92} for {\small QE}. 
In Sections~\ref{sec:FPS} and \ref{sect:Algorithm}
we turn our attention to algorithms
that are more suitable for implementation even though
they rely on sub-algorithms with a doubly exponential 
running time w.r.t.\ $d$.

\smallskip\noindent{\small \bf A special case of quantifier elimination.}
By means of triangular decomposition of algebraic sets over {\C},
triangular decomposition
of semi-algebraic systems  (both full and lazy)  reduces
to a special case of {\small QE}.
In Section~\ref{sec:FPS}, we implement this latter step 
via the concept of a {\em fingerprint polynomial set}, 
which is inspired by that of a {\em discrimination polynomial set}
used for {\small RRC} in~\cite{yhx01,Xiao09}.

\smallskip\noindent{\small \bf Implementation and experimental results.}
In Section~\ref{sect:Algorithm} we describe the algorithms
that we have implemented for computing triangular decompositions
(both full and lazy) 
of semi-algebraic systems.
Our {\sc Maple} code is written
on top of the {\tt RegularChains} library.
We provide experimental data for two groups of well-known problems. 
In the first group, each input semi-algebraic 
system consists of equations only
while the second group is a collection of {\small QE} problems.
To illustrate the difficulty of our test problems,
and only for this purpose, 
we provide timings obtained with other well-known 
polynomial system solvers which are based on algorithms
whose running time estimates are comparable to ours.
For this first group we use the {\sc Maple}
command {\tt Groebner:-Basis} for computing
lexicographical Gr\"obner bases. 
For the second group we use 
a general purpose {\small QE} software:
{\sc qepcad b} (in its non-interactive mode)~\cite{Bro03}.
Our experimental results show that our {\sf LazyRealTriangularize}
code can solve most of our test problems 
and that it can solve more problems than the 
package it is compared to.

We conclude this introduction by 
 computing a 
triangular decomposition of a particular semi-algebraic system
taken from ~\cite{Brown05}.
Consider the following question: when does 
$p(z)=z^3+az+b$ have a non-real root $x+iy$
satisfying $xy<1$~?
This problem can be expressed as
$(\exists x)(\exists y)[f=g=0\wedge y\neq0\wedge xy-1<0]$, where
$f={\rm Re}(p(x+iy)) = x^3-3xy^2+ax+b$
and
$g = {\rm Im}(p(x+iy))/y = 3x^2-y^2+a$.

We call our {\sf LazyRealTriangularize} command
on the semi-algebraic system 
$f = 0, g = 0, y\neq0, xy-1<0$
with the variable order $y>x>b>a$.
Its first step is to call the {\sf Triangularize} 
command of the {\tt RegularChains}
library on the algebraic system $f =  g = 0$.
We obtain one squarefree regular chain $T = [t_1, t_2]$, 
where $t_1=g$ and $t_2 = 8x^3+2ax-b$, satisfying $V(f, g) = V(T)$.
The second step of {\sf LazyRealTriangularize}
is to check whether the polynomials
defining inequalities and inequations are regular w.r.t.\
the saturated ideal of $T$, which is the case here.
The third step is to compute the so called
{\em border polynomial set} (see Section~\ref{sect:preliminary})
which is 
$B=[h_1,h_2]$ with 
$h_1=4a^3+27b^2$
and
$h_2=-4a^3b^2-27b^4+16a^4+512a^2+4096$.
One can check that the regular system $[T, \{y, xy-1\}]$ 
specializes well outside of the hypersurface $h_1 h _2 = 0$.
The fourth step is to compute the fingerprint polynomial set
which yields the quantifier-free formula $\QQ = h_1>0$
telling us that $[\QQ, T, 1-xy > 0]$ is a regular semi-algebraic system.
After performing these four steps,
(based on 
Algorithm~\ref{Algo:LazyRealTriangularize}, 
Section~\ref{sect:Algorithm}) 
the function call
$$
\LazyRealTriangularize{[f, g, y\neq0, xy-1<0], [y, x, b, a]}
$$ 
in our implementation returns the following:
{\small
$$
\left\{
\begin{array}{ll}
[\{t_1= 0, t_2 = 0, 1-xy>0, h_1>0\}] & h_1h_2\neq0\\
  &\\
\mbox{{\tt \%LazyRealTriangularize}}
([t_1=0, t_2=0, f=0,&\\
 h_1=0, 1-xy>0, y\neq0],[y, x, b, a]) & h_1=0\\
  &\\
\mbox{{\tt \%LazyRealTriangularize}}
([t_1=0, t_2=0, f=0,&\\
 h_2=0, 1-xy>0, y\neq0],[y, x, b, a]) & h_2=0\\
\end{array}
\right.
$$
}
The above output shows that $\{ [\QQ, T, 1-xy > 0] \}$ 
forms a lazy triangular decomposition of the input semi-algebraic system.
Moreover, together with the output of the recursive calls,
one obtains  a full triangular decomposition.
Note that the cases of the two recursive calls
correspond to $h_1=0$ and  $h_2=0$.
Since our {\sf LazyRealTriangularize} uses the 
{\sc Maple} piecewise structure for formatting its output,
one simply needs to evaluate 
the recursive calls with the \verb+value+ command, yielding 
the same result as directly calling {\sf RealTriangularize}

{\small
$$
\left\{
\begin{array}{ll}
 [\{t_1= 0, t_2 = 0, 1-xy>0, h_1>0\}] & h_1h_2\neq0\\
  &\\
 {[~]}& h_1=0\\
  &\\
 \left\{
 \begin{array}{ll}
  ([\{t_3= 0, t_4 = 0, &\\
 h_2=0, 1-xy>0\}])&h_3\neq0\\
  &\\
 {[~]}&h_3=0
 \end{array}
 \right.                                                       & h_2=0
\end{array}
\right.
$$
}
where $t_3 =xy+1$, $t_4 =2a^3x-a^2b+32ax-48b+18xb^2$,
$h_3  = (a^2+48)(a^2+16)(a^2+12)$.

From this output, after some simplification, one could obtain the
equivalent quantifier-free formula, $4a^3+27b^2>0$, of the original {\small QE} problem.


\section{Triangular decomposition  of \\ Algebraic Sets}
\label{sect:preliminary}
We review in the section the basic notions related
to regular chains and triangular decompositions
of algebraic sets.
Throughout this paper, let $\K$ be a field of 
characteristic 0 and
$\KK$ be its algebraic closure.
Let $\K[\x]$ be the polynomial ring over 
$\K$ and with ordered variables $\x= x_1 < \cdots < x_n$. 
Let $p,q \in {\K}[\x]$ be polynomials. 
Assume that $p\notin\K$.
Then denote by
$\mvar{p}$, $\init{p}$, and $\mdeg{p}$ 
 respectively
the greatest variable appearing in $p$ 
(called the {\em main variable} of $p$),
the leading coefficient of $p$ w.r.t.\ $\mvar{p}$
(called the {\em initial} of $p$), and
the degree of $p$ w.r.t.\ $\mvar{p}$
(called the {\em main degree} of $p$);
denote by $\der{p}$ the derivative of $p$ w.r.t.\ $\mvar{p}$;
denote by $\discrim{p}$ the discriminant of $p$
w.r.t.\ $\mvar{p}$.

\smallskip\noindent{\small \bf Triangular sets}.
Let $T \subset {\K}[\x]$ be a {\em triangular set},
that is, a set of non-constant polynomials with pairwise
distinct main variables. 
Denote by $\mvar{T}$ the set of
main variables of the polynomials in $T$. 
A variable $v$ in $\x$ is called
{\em algebraic} w.r.t.\ $T$ if $v\in\mvar{T}$, 
otherwise it is said {\em free} w.r.t.\ $T$. 
If no confusion is possible, we shall always denote by
$\u=u_1,\ldots,u_d$ and $\y=y_1,\ldots,y_m$ 
respectively the free and the main variables of $T$.
Let $h_T$ be the product
of the initials of the polynomials in $T$.
We denote by \sat{T} the {\em saturated ideal} of $T$: 
if $T$ is the empty triangular set, then \sat{T} is defined as the trivial
ideal $\langle 0 \rangle$, otherwise it is the ideal
$\langle T \rangle:h_{T}^{\infty}$.
The {\em quasi-component} $W(T)$ of $T$
is defined as $V(T) \setminus V(h_T)$.
Denote $\overline{W(T)} = V(\sat{T})$ as the Zariski closure of $W(T)$.

\smallskip\noindent{\small \bf Iterated resultant}.
Let $p$ and $q$ be two polynomials of $\K[\x]$. 
Assume $q$ is non-constant and let $v=\mvar{q}$.
We define $\res(p, q, v)$ as follows:
if $v$ does not appear in $p$, 
then $\res(p, q, v):=p$; otherwise $\res(p, q, v)$
is the resultant of $p$ and $q$ w.r.t.\ $v$.
Let $T$ be a triangular set of $\K[\x]$.
We define $\res(p, T)$ by induction:
if $T$ is empty, then $\res(p, T)=p$;
otherwise let $v$ be the greatest variable appearing in $T$, then 
$\res(p, T)=\res(\res(p, T_v, v), T_{<v})$, 
where $T_v$ and $T_{<v}$ denote respectively the polynomials
of $T$ with main variables equal to and less than $v$.

\smallskip\noindent{\small \bf Regular chain}.
A triangular set $T \subset \K[\x]$ is a {\em regular chain} 
if: either $T$ is empty; 
or (letting $t$ be  the polynomial in $T$ with maximum main variable),
 $T \setminus \{ t\}$ is a regular chain,
and the initial of $t$ is regular w.r.t.\ 
\sat{T \setminus \{ t \}}.
The empty regular chain is denoted by $\varnothing$.
Let $H\subset\K[\x]$.
The pair $[T, H]$ is a {\em regular system} if each polynomial in $H$
is regular modulo $\sat{T}$. 
A regular chain $T$ or a regular system $[T,H]$, is {\em squarefree} 
if for all $t\in T$, the $\der{t}$ is regular w.r.t.\ $\sat{T}$.

\smallskip\noindent{\small \bf Triangular decomposition}.
Let $F\subset\K[\x]$. Regular chains $T_1,\ldots,T_e$ of $\K[\x]$
form a {\em triangular decomposition} of $V(F)$ if either: 
$V(F)=\cup_{i=1}^e \overline{W(T_i)}$ (Kalkbrener's sense) or 
$V(F)=\cup_{i=1}^e W(T_i)$ (Lazard's sense).
In this paper, we denote by {\sf Triangularize} an algorithm,
such as the one of~\cite{MMM99}, 
computing a triangular decomposition of the former kind.

\smallskip\noindent{\small \bf Regularization.}
Let $p\in\K[\x]$.
Let $T$ be a regular chain of $\K[\x]$.
Denote by $\Regularize{p, T}$ an operation which computes  
a set of regular chains $\{T_1,\ldots, T_e\}$ such that 
$(1)$ for each $i, 1\leq i\leq e$, either $p\in\sat{T_i}$ or $p$ is
regular w.r.t.\ $\sat{T_i}$; 
$(2)$ 
we have $\overline{W(T)}=\overline{W(T_1)}\cup\cdots\cup\overline{W(T_e)}$,
$\mvar{T}=\mvar{T_i}$ and $\sat{T}\subseteq\sat{T_i}$ for $1\leq i\leq e$.

\smallskip\noindent{\small \bf Good specialization~\cite{CGLMP07}}. 
Consider a squarefree regular system $[T, H]$ of $\K[\u,\y]$.
Recall that $\y$ and $\u =  u_1, \ldots, u_d$ 
stand respectively for $\mvar{T}$ and $\x \setminus \y$.
Let $z = (z_1, \ldots, z_d)$ be a point of ${\KK}^d$.
We say that $[T, H]$ {\em specializes well} at $z$ if:
$(i)$ none of the initials of the polynomials in $T$ vanishes
modulo the ideal $\langle z_1 - u_1, \ldots, z_d - u_d \rangle$;
$(ii)$ the image of $[T, H]$ modulo 
              $\langle z_1 - u_1, \ldots, z_d - u_d \rangle$
              is a squarefree regular system.

\smallskip\noindent{\small \bf Border polynomial~\cite{yhx01}}. 
Let $[T,H]$ be a squarefree regular system of $\K[\u,\y]$.
Let $bp$ be the primitive and square free part of 
the product of all
$\res(\der{t}, T)$ and all $\res(h, T)$ for $h\in H$ and $t\in T$.
We call $bp$ the {\em border polynomial} of $[T, H]$
and denote by ${\sf BorderPolynomial}(T,H)$ an algorithm to compute it. 
We call the set of irreducible factors of $bp$ the
{\em border  polynomial set} of $[T, H]$.
Denote by ${\sf BorderPolynomialSet}(T, H)$ an algorithm
to compute it.
Proposition~\ref{Proposition:borderpolynomial}
follows from 
the specialization property of subresultants
and states a fundamental property of border polynomials.

\begin{Proposition}
\label{Proposition:borderpolynomial}
The system $[T, H]$ specializes well at $u\in\KK^d$ if and only if 
the border polynomial $bp(u)\neq0$.
\end{Proposition}


\section{Triangular decomposition of \\
semi-algebraic systems}
\label{sect:specification}
In this section, we  prove that any semi-algebraic system can be decomposed 
into finitely many regular semi-algebraic systems.
This latter notion was defined in the introduction.

\smallskip\noindent{\small \bf Semi-algebraic system}. 
Consider four finite polynomial subsets
$F=\{f_1,f_2,\cdots,f_s\}$,
$N = \{n_1,n_2,\cdots,n_t\}$,
$P=\{p_1,p_2,\cdots,p_r\}$,
and $H=\{h_1,h_2,\cdots,h_{\ell}\}$ 
of ${\Q}[x_1, \ldots, x_n]$.
Let $\NN$ denote the set of non-negative inequalities 
$\{n_1 \geq 0, \ldots, n_t \geq 0\}$.
Let $\PP$ denote the set of positive inequalities
$\{p_1 > 0, \ldots, p_r > 0\}$.
Let $\HH$ denote the set of inequations $\{h_1\neq0,\ldots, h_{\ell}\neq0\}$.
We will denote by $[F,\PP]$ the {\em basic semi-algebraic system}
$\{f_1 = 0, \ldots, f_s = 0, p_1 > 0, \ldots, p_r  > 0\}$.
We denote by $\S=[F,\NN,\PP,\HH]$
the semi-algebraic system ({\small SAS}) 
which is the conjunction of the following conditions:
$f_1 = 0, \ldots, f_s = 0$,
$n_1 \geq 0, \ldots, n_t \geq 0$,
$p_1 > 0, \ldots, p_r > 0$ and
$h_1\neq 0, \ldots, h_{\ell} \neq 0$.

\smallskip\noindent{\small \bf Notations on zero sets}. 
In this paper, we use ``$Z$'' to denote the zero set of 
a polynomial system, involving equations and inequations, in $\C^n$ 
and ``$Z_{\R}$'' to denote the zero set
of a semi-algebraic system in $\R^n$.

\smallskip\noindent{\small \bf Pre-regular semi-algebraic system}.
Let $[T, P]$ be a squarefree regular system  of $\Q[\u, \y]$.
Let $bp$ be the border polynomial of $[T, P]$.
Let $B\subset\Q[\u]$ be a polynomial set such that 
$bp$ divides the product of polynomials in $B$.
We call the triple $[B_{\neq}, T, \PP]$ a 
{\em pre-regular semi-algebraic system} of $\Q[\x]$.
Its zero set, written as $Z_{\R}(B_{\neq}, T, \PP)$, is 
the set $(u, y)\in\R^n$ such that $b(u)\neq0$, $t(u,y)=0$, $p(u, y)>0$,
for all $b\in B$, $t\in T$, $p\in P$.
Lemma~\ref{Lemma:decompose-prsas} and Lemma~\ref{Lemma:prsas}
are fundamental properties of pre-regular semi-algebraic systems.

\begin{Lemma}
\label{Lemma:decompose-prsas}
Let $\S$ be a semi-algebraic system of $\Q[\x]$.
Then there exists finitely many pre-regular semi-algebraic systems
$[{B_i}_{\neq}, T_i, {P_i}_>]$, $i = 1 \cdots e$, 
s.t. $Z_{\R}(\S)=\cup_{i=1}^e Z_{\R}({B_i}_{\neq}, T_i, {P_i}_>)$.
\end{Lemma}
\begin{proof}
The semi-algebraic system $\S$ decomposes
into basic semi-algebraic systems,
by rewriting  inequality of type $n\ge0$ as: $n>0 \, \vee \, n=0$.
Let $[F, \PP]$ be one of those basic semi-algebraic systems.
If $F$ is empty, then the triple $[P, \varnothing, \PP]$, 
is a pre-regular semi-algebraic system.
If $F$ is not empty, by Proposition~\ref{Proposition:borderpolynomial}
and the specifications of {\sf Triangularize} and {\sf Regularize}, 
one can compute finitely many
squarefree regular systems $[T_i, H]$ 
such that $
V(F)\cap Z(P_{\neq})=\cup_{i=1}^e \left(V(T_i)\cap Z({B_i}_{\neq})\right)
$ holds and 
where $B_i$ is the border polynomial set of the regular system $[T_i, H]$.
Hence, we have
$
Z_{\R}(F, \PP)=\cup_{i=1}^e Z_{\R}({B_i}_{\neq}, T_i, \PP),
$
where each 
$[{B_i}_{\neq}, T_i, \PP]$ is a pre-regular semi-algebraic system. 
\end{proof}

\begin{Lemma}
\label{Lemma:prsas}
Let $[B_{\neq}, T, \PP]$ be a pre-regular semi-algebraic system 
of $\Q[\u,\y]$.
Let $h$ be the product of polynomials in $B$.
The complement of the hypersurface $h=0$ in $\R^d$ consists of 
finitely many open cells 
of dimension $d$.
Let $C$ be one of them.
Then for all $\alpha\in C$, 
the number of real zeros of $[T(\alpha), \PP(\alpha)]$ is the same.
\end{Lemma}
\begin{proof}
From 
Proposition~\ref{Proposition:borderpolynomial}
and recursive use of Theorem 1 in~\cite{col75}
on the delineability of a polynomial.
\end{proof}

\begin{Lemma}
\label{Lemma:pre2regular}
Let $[B_{\neq}, T, \PP]$ be a pre-regular semi-algebraic system 
 of $\Q[\u, \y]$.
One can decide whether its zero set is empty or not.
If it is not empty, then one can compute 
a regular semi-algebraic system $[\QQ, T, \PP]$
whose zero set in $\R^n$ is the same as that of $[B_{\neq}, T, \PP]$.
\end{Lemma}
\begin{proof}
If $T=\varnothing$,
we can always test whether 
the zero set of $[B_{\neq},\PP]$ is empty or not,
for instance using {\small CAD}.
If it is empty, we are done.
Otherwise, defining $\QQ=B_{\neq}\wedge\PP$, 
the triple $[\QQ, T, \PP]$ is a regular semi-algebraic system.
If $T$ is not empty, 
we solve the quantifier elimination problem 
$\exists\y(B(\u)\neq 0, T(\u, \y)=0, P(\u, \y)>0)$ and let
$\QQ$ be the resulting formula.
If $\QQ$ is false, we are done.
Otherwise, by Lemma~\ref{Lemma:prsas},
above each connected component of $B(\u)\neq 0$, 
the number of real zeros of the system $[B_{\neq}, T, \PP]$
is constant.
Then, the zero set defined by $\QQ$ is the union
of the connected  components of 
$B(\u)\neq 0$ above which $[B_{\neq}, T, \PP]$
possesses at least one solution.
Thus, $\QQ$
defines a nonempty open set of $\R^d$
and 
$[\QQ, T, \PP]$ is a regular semi-algebraic system.
\end{proof}

\begin{Theorem}
\label{Theorem:decompose}
Let $\S$ be a semi-algebraic system of $\Q[\x]$.
Then one can compute a (full) triangular decomposition
of $\S$, that is, as defined  in the introduction, 
finitely many regular semi-algebraic 
systems such that the union of their zero sets is the 
zero set of $\S$.
\end{Theorem}
\begin{proof}
It follows from Lemma~\ref{Lemma:decompose-prsas}
and~\ref{Lemma:pre2regular}.
\end{proof}

\section{Complexity Results}
\label{sec:BPRTD}

We start this section by stating complexity 
estimates for basic operations on multivariate polynomials.

\smallskip\noindent{\small \bf Complexity of basic polynomial operations.} 
Let $p, q \in {\Q}[\x]$ be polynomials with
respective total degrees ${\delta}_p, {\delta}_q$,
and let $x\in \x$.
Let $\hbar_p, \hbar_q, \hbar_{pq}$ and $ \hbar_r$
be the {\em height}
(that is, the bit size of the maximum absolute value of 
the numerator or denominator of a coefficient) of $p,q$,
the product $pq$ and the resultant $\res(p,q,x)$,  respectively.
In~\cite{DST88}, it is proved that 
$\gcd(p,q)$ can be computed within $O(n^{2\delta+1}\hbar^3)$ bit operations
where ${\delta} = {\max}({\delta}_p, {\delta}_q)$
and $\hbar = {\max}(\hbar_p, \hbar_q)$.
It is easy to establish that 
$\hbar_{pq}$ and $\hbar_r$ are respectively
upper bounded by 
$\hbar_p +\hbar_q + n\log(\min(\delta_p,\delta_q)+1)$
and 
$\delta_q\hbar_p +\delta_p\hbar_q+n\delta_q\log(\delta_p+1)+n\delta_p\log(\delta_q+1)+\log\left((\delta_p+\delta_q)!\right)$.
Finally, let 
$M$ be a $k \times k$ matrix over ${\Q}[\x]$.
Let ${\delta}$ (resp. $\hbar$)
be the maximum total degree (resp. height)
of a polynomial coefficient of $M$.
Then $\det(M)$ can be computed within 
$O(k^{2n+5}(\delta+1)^{2n}\hbar^2)$ 
bit operations, see~\cite{HS98}.

We turn now to the main subject of this section,
that is, complexity estimates for 
a lazy triangular decomposition of a polynomial system
under some genericity assumptions.
Let $F \subset {\Q}[\x]$. 
A lazy triangular decomposition (as defined in the Introduction) of 
the semi-algebraic system $\S=[F,\emptyset,\emptyset,\emptyset]$, 
which only involves equations,
is obtained by the above
algorithm.

\begin{algorithm}
\dontprintsemicolon
\linesnumbered
\caption{\LazyRealTriangularize{\S}
\label{Algo::3-step}}
\KwIn{a semi-algebraic system $\S=[F, \emptyset, \emptyset, \emptyset]$}
\KwOut{a lazy triangular decomposition of $\S$}
$\T:=\Triangularize{F}$\; 
\For{$T_i \in\T$}{
$bp_i:=\BP{T_i,\emptyset}$ \; 
solve $\exists\y(bp_i(\u)\neq 0, T_i(\u, \y)=0)$, 
and let $\QQ_i$ be the resulting quantifier-free formula\;
\lIf{$\QQ_i \neq false$}{
output $[\QQ_i, T_i, \emptyset]$
}
}
\end{algorithm}

\smallskip\noindent{\small \bf Proof of Algorithm~\ref{Algo::3-step}}.
The termination of the algorithm is obvious.
Let us prove its correctness.
Let $R_i=[\QQ_i, T_i, \emptyset]$, for $i=1\cdots t$ be 
the output of Algorithm~\ref{Algo::3-step} 
and let $T_j$ for $j=t+1\cdots s$ be the
regular chains such that $\QQ_j=false$. 
By Lemma~\ref{Lemma:pre2regular}, each $R_i$ is a regular semi-algebraic system.
For $i=1\cdots s$, define $F_i=\sat{T_i}$.
Then we have
$V(F) = \cup_{i=1}^s V(F_i)$, 
where each $F_i$ is equidimensional. 
For each $i = 1\cdots s$, by Proposition~\ref{Proposition:borderpolynomial}, 
we have $$V(F_i)\setminus V(bp_i)=V(T_i) \setminus V(bp_i).$$
Moreover, we have 
$$V(F_i)=\left(V(F_i)\setminus V(bp_i)\right) \cup V(F_i\cup\{bp_i\}).$$
Hence,
$$Z_{\R}(R_i)=Z_{\R}(T_i)\setminus Z_{\R}(bp_i)\subseteq Z_{\R}(F_i)\subseteq Z_{\R}(F)$$ holds.
In addition, since $bp_i$ is regular modulo $F_i$, we have
$$
\begin{array}{rcl}
Z_{\bb{R}}(F) \setminus \cup_{i=1}^tZ_{\R}(R_i)
&=& \cup_{i=1}^s Z_{\R}(F_i)\setminus \cup_{i=1}^t Z_{\R}(R_i)\\
&\subseteq& \cup_{i=1}^s Z_{\R}(F_i) \setminus \left(Z_{\R}(T_i)\setminus Z_{\R}(bp_i)\right)\\
&\subseteq& \cup_{i=1}^s Z_{\R}(F_i \cup  \{bp_i\}),
\end{array}
$$ and $\dim\left(\cup_{i=1}^s V(F_i \cup  \{bp_i\})\right) <  \dim(V(F))$. 
So the $R_i$, for $i=1\cdots t$, 
form a lazy triangular decomposition of $\S$. 
$\square$

In this section, under some genericity assumptions for $F$,
we establish running time estimates for 
Algorithm~\ref{Algo::3-step},
see Proposition~\ref{prop:lazy-rtd}.
This is achieved through:
\begin{itemizeshort}
\item[ $(1)$  ] Proposition~\ref{Prop:foobar}  
giving running time and output size estimates
for a Kalkbrener triangular decomposition of an algebraic set, and
\item[ $(2)$  ] Theorem~\ref{thm: complexity-bp}
 giving running time and output size estimates
for a border polynomial computation.
\end{itemizeshort}
Our assumptions for these results are the following:
\begin{itemizeshort}
\item[ $\mathbf{(H_0)}$] $V(F)$ is equidimensional of dimension $d$,
\item[ $\mathbf{(H_1)}$] $x_1, \ldots, x_d$ are algebraically independent 
   modulo each associated prime ideal of 
   the ideal generated by $F$ in ${\Q}[\x]$,
\item[ $\mathbf{(H_2)}$] $F$ consists of $m := n - d$ polynomials,
      $f_1, \ldots, f_m$.
\end{itemizeshort}
Hypotheses $\mathbf{(H_0)}$ and $\mathbf{(H_1)}$
are equivalent to the existence of  
regular chains $T_1, \ldots, T_e$ of ${\Q}[x_1, \ldots, x_n]$
such that $x_1, \ldots, x_d$ are free
w.r.t.\ each of $T_1, \ldots, T_e$ and such that we have
$V(F) = \overline{W(T_1)} \, \cup \, \cdots \, \cup \, \overline{W(T_e)}.$

Denote by ${\delta}$, $\hbar$ respectively 
the maximum total degree and height 
of $f_1, \ldots, f_m$.
In her PhD Thesis~\cite{Szanto99}, {\'A}.~Sz{\'a}nt{\'o}
describes an algorithm
which computes a Kalkbrener 
triangular decomposition, $T_1, \ldots, T_e$,  of $V(F)$.
Under Hypotheses $\mathbf{(H_0)}$
to $\mathbf{(H_2)}$, this algorithm runs in time
$m^{O(1)} ({\delta}^{O(n^2)})^{d + 1}$ counting operations in {\Q},
while the total degrees of the 
polynomials in the output are bounded by $n{\delta}^{O(m^2)}$.
In addition, $T_1, \ldots, T_e$ are square free, 
{\em strongly normalized}~\cite{MMM99} and {\em reduced}~\cite{ALM99}.

From $T_1, \ldots, T_e$, we obtain 
regular chains $E_1, \ldots, E_e$ forming another
Kalkbrener triangular decomposition of $V(F)$, as follows.
Let $i= 1 \cdots e$ and $j = (d+1) \cdots n$.
Let $t_{i,j}$ be the polynomial of $T_i$ with $x_j$ as main variable.
Let ${e_{i,j}}$ be the primitive part of $t_{i,j}$ 
regarded as a polynomial in ${\Q}[x_1,\ldots,x_d][x_{d+1}, \ldots, x_n]$.
Define $E_i = \{ e_{i,d+1}, \ldots, e_{i,n} \}$.
According to the complexity results for polynomial
operations stated at the beginning of this section, 
this transformation can be done 
within ${\delta}^{O(m^4) O(n)}$ operations in {\Q}.

Dividing ${e_{i,j}}$ by its initial
we obtain a monic polynomial $d_{i,j}$ 
of ${\Q}(x_1, \ldots, x_d)[x_{d+1}, \ldots, x_n]$.
Denote by $D_i$ the regular chain $\{ d_{i,d+1}, \ldots, d_{i,n} \}$.
Observe that $D_i$ is the 
reduced lexicographic Gr\"obner basis of the radical ideal 
it generates in 
${\Q}(x_1, \ldots,\allowbreak x_d)[x_{d+1}, \ldots, x_n]$.
So Theorem 1 in~\cite{DaSKac09} applies to each regular chain $D_i$.
For each polynomial $d_{i,j}$, this theorem provides 
height and total degree estimates
expressed as functions of the {\em degree}~\cite{BCSL97}  and
the {\em height}~\cite{pph3,KrPaSo01} of the algebraic set
$\overline{W(D_i)}$.
Note that the degree and height of
$\overline{W(D_i)}$ are upper bounded by those of $V(F)$.
Write $d_{i,j} = {\Sigma}_{\mu} \, \frac{{\alpha}_{\mu}}{{\beta}_{\mu}} \, {\mu}$
where each ${\mu} \in {\Q}[x_{d+1}, \ldots, x_n]$ is a monomial
and ${\alpha}_{\mu}, {\beta}_{\mu}$ are in 
${\Q}[x_1, \ldots,\allowbreak x_d]$ such that 
${\gcd}({\alpha}_{\mu}, {\beta}_{\mu}) = 1$ holds.
Let ${\gamma}$ be the lcm of the ${\beta}_{\mu}$'s.
Then for ${\gamma}$ and each ${\alpha}_{\mu}$:
\begin{itemizeshort}
	\item the total degree is bounded by $2 {\delta}^{2m}$ and,
	\item  the height 
           by $O({\delta}^{2m}(m \hbar +  d m \, {\log}({\delta}) 
                    + n {\log}(n)))$.
\end{itemizeshort}
Multiplying $d_{i,j}$ by ${\gamma}$
brings $e_{i,j}$ back. We deduce
the height and total degree 
estimates for each $e_{i,j}$ below.

\begin{Proposition}
\label{Prop:foobar}
The Kalkbrener triangular decomposition $E_1, \ldots, E_e$ 
of $V(F)$ can be computed in
${\delta}^{O(m^4) O(n)}$ operations in {\Q}.
In addition,  every polynomial $e_{i,j}$
has total degree upper bounded by $4 {\delta}^{2m} + {\delta}^m$,
and has height upper bounded by
$O({\delta}^{2m}(m \hbar +  d m  {\log}({\delta}) 
                    + n {\log}(n)))$.
\end{Proposition}

Next we estimate the running time
and output size for computing the border polynomial
of a regular system.
\begin{Theorem}
\label{thm: complexity-bp}
Let $R=[T,P]$ be a squarefree regular system of $\Q[\u,\y]$,
with $m = \# T$ and ${\ell} = \# P$.
Let $bp$ be the border polynomial of $R$.
Denote by ${\delta}_R$, ${\hbar}_R$ 
respectively the maximum total degree and height
of a polynomial in $R$.
Then the total degree of $bp$ is upper bounded by 
$(\ell+m)2^{m-1}{{\delta}_R}^{m}$, 
and $bp$ can be computed within
$(n\ell+nm)^{O(n)} (2{{\delta}_R})^{O(n)O(m)}{{\hbar}_R}^3$ 
bit operations.
\end{Theorem}
\begin{proof}
Define $G := P \cup \{ \der{t} \, \mid \, t \in T  \}$.
We need to compute the $\ell+m$ iterated resultants $\res(g,T)$,
for all $g \in G$.
Let $g \in G$.
Observe that the total degree and height of $g$
are  bounded by 
${{\delta}_R}$ and ${{\hbar}_R}+\log({{\delta}_R})$ respectively.
Define 
$r_{m+1}:=g$, \ldots, 
$r_i:=\res(t_i,r_{i+1},y_i)$, \ldots,
$r_1 :=\res(t_1,r_{2},y_1)$.
Let $i \in \{ 1, \ldots, m \}$.
Denote by ${\delta}_i$ and  ${{\hbar}_i}$ the total degree and
height of $r_i$, respectively.
Using the complexity estimates stated at the beginning 
of this section, we have
${{\delta}_i} \leq 2^{m-i+1}{{\delta}_R}^{m-i+2}$ 
and
${{\hbar}_i} \leq 2{{\delta}}_{i+1}({{\hbar}_{i+1}} +n\log({{\delta}}_{i+1}+1))$.
Therefore, we have 
${{\hbar}}_{i}\leq (2{{\delta}_R})^{O(m^2)}n^{O(m)}{{\hbar}_R}$.
From these size estimates, one can deduce that each resultant $r_i$ 
(thus the iterated resultants) can be computed within 
$(2{{\delta}_R})^{O(mn)+O(m^2)}n^{O(m)}{{\hbar}_R}^2$ bit operations, 
by the complexity of computing a determinant 
stated at the beginning of this section.

Hence, the product of all iterated resultants has total
degree and height bounded
by $(\ell+m)2^{m-1}{{\delta}_R}^{m}$ 
and
$(\ell+m)(2{{\delta}_R})^{O(m^2)}n^{O(m)}{{\hbar}_R}$, respectively. 
Thus, the primitive and squarefree part of this product 
can be computed within 
$(n\ell+nm)^{O(n)} (2{{\delta}_R})^{O(n)O(m)}{{\hbar}_R}^3$ bit operations, 
based on the complexity of  a polynomial
gcd computation stated at the beginning of this section.
\end{proof}

\begin{Proposition}
		\label{prop:lazy-rtd}
From the Kalkbrener triangular decomposition $E_1, \ldots, E_e$ of  
Proposition~\ref{Prop:foobar}, 
a lazy triangular decomposition of $f_1 = \cdots = f_m =0$
can be computed
in $\left({\delta^{n^2}n4^n}\right)^{O(n^2)}\hbar^{O(1)}$  bit operations.
Thus, a lazy triangular decomposition of this system 
is computed from the input polynomials in singly exponential time
w.r.t.\ $n$, counting operations in {\Q}.
\end{Proposition}
\begin{proof}
For each $i \in \{1 \cdots e\}$, 
let $bp_i$ be the border polynomial of $[E_i,\emptyset]$ and
let $\hbar_{R_i}$ (resp. $\delta_{R_i}$) be the height 
(resp. the total degree) bound
of the polynomials in the pre-regular 
semi-algebraic system $R_i=[\{bp_i\}_{\neq}, E_i, \emptyset]$.
According to Algorithm~\ref{Algo::3-step}, the remaining task is to solve 
the {\small QE} problem $\exists\y(bp_i(\u)\neq 0, E_i(\u, \y)=0)$ for
each $i \in \{1 \cdots e\}$, 
which can be solved within 
$\left((m+1)\delta_{R_i}\right)^{O(dm)}\hbar_{R_i}^{O(1)}$ bit operations,
based on the results of~\cite{Ren92}.
The conclusion 
 follows from the size estimates in Proposition \ref{Prop:foobar} and
Theorem \ref{thm: complexity-bp}. 
\end{proof}

\section{Quantifier Elimination by RRC}
\label{sec:FPS}

In the last two sections, we saw that in order 
to compute a triangular decomposition of 
a semi-algebraic system, 
a key step is to solve the following quantifier 
elimination problem: 
\begin{equation}
\label{eq:RRC0}
\exists\y(B(\u)\neq 0, T(\u, \y)=0, P(\u, \y)>0),
\end{equation}
where $[B_{\neq}, T, \PP]$ is a pre-regular semi-algebraic 
system of $\Q[\u, \y]$. 
This problem is an instance of the so-called
{\em real root classification} ({\small RRC})~\cite{YX05}.
In this section, we show how to solve this problem
when $B$ is what we call a {\em fingerprint polynomial set}.

\smallskip\noindent{\small \bf Fingerprint polynomial set.}
Let $R:=[B_{\neq}, T, \PP]$ be a pre-regular semi-algebraic system of $\Q[\u,\y]$.
Let $D \subset \Q[\u]$.
Let $dp$ be
the product of all polynomials in $D$.
We call  $D$  a {\em fingerprint polynomial set} ({\small FPS}) of $R$ if:
\begin{itemizeshort}
\item[$(i)$] for all $\alpha\in\R^d$, for all $b\in B$ we have:\\
  $~~~~~dp(\alpha)\neq0 \Longrightarrow b(\alpha)\neq0$,
\item[$(ii)$] for all ${\alpha}, {\beta} \in \R^d$ with ${\alpha} \neq {\beta}$
and $dp(\alpha)\neq0$, $dp(\beta)\neq0$, 
if the signs of $p(\alpha)$  and $p(\beta)$ are 
the same for all $p \in D$, then 
$R(\alpha)$ has real solutions if and only if $R(\beta)$ does.
\end{itemizeshort}

Hereafter,
we present a method to construct an {\small FPS}
based on projection operators of {\small CAD}.

\smallskip\noindent{\small \bf Open projection operator~\cite{adam00, brown01}.}
Hereafter in this section, let $\u=u_1<\cdots<u_d$ be ordered variables.
Let $p \in \Q[\u]$ be non-constant.
Denote by $\factor{p}$ the set
of the non-constant irreducible factors of $p$.
For $A\subset\Q[\u]$, define $\factor{A} = {\cup}_{p \in A} \, \factor{p}$.
Let ${C}_{d}$ (resp. $C_0$) be the set of the
polynomials in $\factor{p}$ with
main variable equal to (resp. less than) $u_d$. 
The {\em open projection operator} (${\rm oproj}$) w.r.t.\ variable $u_d$
maps $p$ to
a set of polynomials of $\Q[u_1,\ldots,u_{d-1}]$ defined below:
$$
\begin{array}{r}
\oproj{p,u_d}:= C_0\cup\bigcup_{f,g \in C_d, \; f\neq g}\factor{\res(f, g, u_d)}\\
              \cup\bigcup_{f\in C_d} \factor{ \init{f, u_d}\cdot\discrim{f, u_d}}.
\end{array}
$$
Then, we define $\oproj{A,u_d}:=\oproj{\Pi_{p\in A} \, p,u_d}$.

\smallskip\noindent{\small \bf Augmentation.}
Let $A \subset \Q[\u]$ and $x \in \{ u_1, \ldots, u_d \}$.
Denote by $\der{A,x}$  the {\em derivative closure} of $A$ w.r.t.\ $x$, 
that is,
$
\der{A,x}:=  {\cup}_{p \in A} \,  \{\mr{der}^{(i)}(p,x)\mid 0\leq i<\deg(p,x)  \}.
$
The \emph{open augmented projected factors} of $A$ is denoted by $\oaf(A)$
and defined as follows. 
Let $k$ be the smallest positive integer such that
$A \subset {\Q}[u_1, \ldots, u_k]$ holds.
Denote by $C$ the set
$\factor{\der{A, u_k}}$; we have
\begin{itemizeshort}
\item if $k = 1$, then $\oaf(A):=C$;
\item if $k>1$,  then 
		$\oaf(A):=C \cup \oaf(\oproj{C,u_k}).$
\end{itemizeshort}

\begin{Theorem}
\label{thm:oap}
Let $A \subset \Q[\u]$ be finite and let ${\sigma}$ be a map from
$\oaf(A)$ to the set of signs $\{ -1, +1 \}$.
Then the set 
$S_d := {\cap}_{p \in \oaf(A)} \, \{ u \in \R^d \mid p(u)  \, {\sigma}(p) > 0 \}$
is either empty or a connected open set in $\R^d$.
\end{Theorem}
\begin{proof}
By induction on $d$.
When $d=1$, the conclusion follows from Thom's Lemma
~\cite{BPR06}.
Assume $d > 1$. 
If $d$ is not the smallest positive integer $k$ such that
$A \subset {\Q}[u_1, \ldots, u_k]$ holds, then
 $S_d$ can be written $S_{d-1} \times {\R}$
 and the conclusion follows by induction.
Otherwise, write  $\oaf(A)$ as $C \cup E$, 
where $C=\factor{\der{A, u_d}}$ and $E=\oaf(\oproj{C,u_d})$.
We have: $E \subset \Q[u_1,\cdots,u_{d-1}]$. 
Denote by $M$ the set 
${\cap}_{p \in E} \,   \{u \in \R^{d-1} \mid p(u) \sigma(p) > 0 \}$.
If $M$ is empty then so is $S_d$ and the conclusion is clear.
From now on assume $M$ not empty. Then, 
by induction hypothesis, $M$ is 
a connected open set in $\R^{d-1}$.
By the definition of the operator ${\rm oproj}$,
the product of the polynomials in $C$ is
delineable over $M$ w.r.t.\ $u_d$.
Moreover,  $C$ is derivative closed (may be empty) w.r.t.\ $u_d$.
Therefore ${\cap}_{p \in \oaf(A)} \, \{ u \in \R^d \mid p(u)  \, {\sigma}(p) > 0 \} \subset M\times \R$
is either empty or a 
connected open set by Thom's Lemma.
\end{proof}

\begin{Theorem}
\label{thm:dpoap}
Let $R:=[B_{\neq}, T, \PP]$ be a pre-regular semi-algebraic system of $\Q[\u,\y]$.
The polynomial set $\oaf(B)$ is a fingerprint polynomial set of $R$.
\end{Theorem}
\begin{proof}
Recall that the border polynomial $bp$ of $[T, P]$
divides the product of the polynomials in $B$. 
We have $\factor{B} \allowbreak \subseteq \oaf(B)$. 
So $\oaf(B)$ satisfies $(i)$ in the definition of {\small FPS}.
Let us prove $(ii)$.
Let $dp$ be the product of the polynomials in $\oaf(B)$.
Let $\alpha, \beta \in \R^d$ 
such that both $dp(\alpha)\neq0$, $dp(\beta)\neq0$ hold 
and the signs of $p(\alpha)$ and $p(\beta)$ are equal for
all $p \in \oaf(B)$.
Then, by Theorem~\ref{thm:oap}, $\alpha$  and $\beta$ 
belong to the same connected component of $dp(\u) \neq 0$,
and thus to the same connected component of $B(\u)\neq0$.
Therefore the number of real solutions of $R(\alpha)$ and that of $R(\beta)$
are the same by Lemma~\ref{Lemma:prsas}.
\end{proof}

From now on, 
let us assume that the set $B$ in 
the pre-regular semi-algebraic system $R=[B_{\neq}, T, \PP]$
is an {\small FPS} of $R$.
We solve the quantifier elimination problem (\ref{eq:RRC0})
in three steps:
$(s_1)$ compute at least one sample point in each connected component of 
the semi-algebraic set defined by $B(\u) \neq 0$;
$(s_2)$ for each sample point ${\alpha}$ such that 
the specialized system $R({\alpha})$ possesses 
real solutions, compute the sign
of $b({\alpha})$ for each $b \in B$;
$(s_3)$ generate the corresponding quantifier-free formulas.

In practice, when the set $B$ is not an {\small FPS}, 
one adds some polynomials from ${\oaf}(B)$,
using a heuristic procedure (for instance one by one)  
until Property $(ii)$ of the definition of an {\small FPS} is
satisfied.
This strategy is implemented in Algorithm~\ref{Algo:GenerateRegularSas}
of Section~\ref{sect:Algorithm}.

\section{Implementation}
\label{sect:Algorithm}
In this section, we present algorithms for {\sf LazyRealTriangularize}
and  {\sf RealTriangularize} that we have implemented 
on top of the {\tt RegularChains} library in {\sc Maple}.
We also provide experimental results
for test problems which are available at
\url{www.orcca.on.ca/~cchen/issac10.txt}.

\begin{algorithm}
\dontprintsemicolon
\linesnumbered
\caption{\GeneratePreRegularSas{\S}\label{Algo:GeneratePreRegularSas}}
\KwIn{a semi-algebraic system $\S=[F, \NN, \PP, \HH]$}
\KwOut{
a set of pre-regular semi-algebraic 
systems\\
$[{B_i}_{\neq}, T_i, {P_i}_{>}]$, $i=1\ldots e$, such that
$
\begin{array}{rcl}
Z_{\R}(\S)&=&\cup_{i=1}^e Z_{\R}({B_i}_{\neq}, T_i,  {P_i}_{>})\\
&&\cup_{i=1}^e Z_{\R}(\sat{T_i}\cup\{\Pi_{b\in B_i}b\}, \NN,  P_{>}, \HH).
\end{array}
$
}
    $\T := \Triangularize{F}$; $\T' := \emptyset$\; 
    \For{$p\in P\cup H$}{
        \For{$T\in\T$}{
            \For{$C\in\Regularize{p, T}$}{
                {\bf if} $p\notin\sat{C}$ {\bf then} $\T':=\T'\cup\{C\}$
            }
        }
        $\T := \T'$; $\T':=\emptyset$\;
    }
    $\T :=\{[T, \emptyset]\mid T\in\T\}$; $\T' := \emptyset$\;
    \For{$p\in N$}{
        \For{$[T, N']\in\T$}{
            \For{$C\in\Regularize{p, T}$}{
                \eIf{$p \in \sat{C}$}{
                     $\T':=\T'\cup\{[C, N']\}$
                   }{
                     $\T':=\T'\cup\{[C, N'\cup\{p\}]\}$
                }
            }
        }
        $\T := \T'$; $\T':=\emptyset$\;
    }
    $\T:=\{[T, N', P, H]\mid [T, N']\in\T\}$\;
    \For{$[T, N', P, H]\in\T$}{
         $B:=\BorderPolynomial{T, N'\cup P\cup H}$\;
         output $[B, T, N'\cup P]$\;
    }
\end{algorithm}

\begin{algorithm}
\dontprintsemicolon
\linesnumbered
\caption{\GenerateRegularSas{B, T, P}\label{Algo:GenerateRegularSas}}
\KwIn{$\S=[B_{\neq}, T, \PP]$, a pre-regular semi-algebraic system of $\Q[\u, \y]$, where $\u=u_1,\ldots,u_d$ and $\y=y_1,\ldots,y_{n-d}$. }
\KwOut{A pair $(D, {\cal R})$ satisfying: \\
$(1)$ $D\subset\Q[\u]$ such that $\factor{B} \subseteq D$;\\
$(2)$ ${\cal R}$ is a finite set of regular semi-algebraic systems, s.t. 
$\cup_{R \in {\cal R} }Z_{\R}(R)=Z_{\R}(D_{\neq}, T, \PP)$.
}
      $D:=\factor{B\setminus \Q}$\; 
      \If{$d=0$}{
         \eIf{$\RealRootCounting{T, P}=0$}{
              return $(D, \emptyset)$\;     
         }{
              return $(D, \{[true, T, P]\})$\;
         }
      }
      \While{true}{
            $S:=\SamplePoints{D, d}$; $G_0:=\emptyset$; $G_1:=\emptyset$\;
            \For{$s\in S$}{
                 \eIf{$\RealRootCounting{T(s), P(s)}=0$}{
                      $G_0:=G_0\cup\{\GenerateFormula{D, s}\}$
                    }{
                      $G_1:=G_1\cup\{\GenerateFormula{D, s}\}$
                 }
            }
            \eIf{$G_0\cap G_1=\emptyset$}{
                 $\QQ:=\ReviseFormula{G_1}$\;
				 {\bf if} $\QQ=false$ {\bf then}  return $(D,\emptyset)$\;
				 {\bf else} return $(D,\{[\QQ, T, P]\})$\;
              }{
                select a subset $D' \subseteq \oaf(B) \setminus D$
                 by some heuristic method\;
                 $D := D \cup D'$\;
            }
      }

\end{algorithm}

\smallskip\noindent{\small \bf Basic subroutines.}
For a zero-dimensional squarefree regular system 
$[T, P]$, \RealRootCounting{T, P}~\cite{xz06} 
returns the number of real zeros of $[T,\PP]$.
For $A\subset\Q[u_1,\ldots,u_d]$ and a point $s$ of $\Q^d$
such that $p(s) \neq 0$ for all $p\in A$, \GenerateFormula{A, s}
computes a formula $\wedge_{p\in A}~(p \, \sigma_{p,s}>\!0)$, where $\sigma_{p,s}$ is defined 
as $+1$ if $p(s)>0$ and $-1$ otherwise.
For a  set of formulas $G$,
\ReviseFormula{G} computes a logic formula $\Phi$ 
equivalent to
the disjunction of the formulas in $G$.

\smallskip\noindent{\small \bf Proof of
Algorithm~\ref{Algo:GeneratePreRegularSas}}.
Its termination is obvious.
Let us prove its correctness.
By the specification of {\sf Triangularize} and {\sf Regularize},
at line $16$, we have
$$
Z(F, P_{\neq}\cup H_{\neq})=\cup_{[T, N', P, H]\in\T} Z(\sat{T}, P_{\neq}\cup \HH).
$$
Write $\cup_{[T, N', P, H]\in\T}$ as $\cup_T$. Then we deduce that
$$
Z_{\R}(F, \NN,\PP,\HH)=\cup_T Z_{\R}(\sat{T}, \NN, \PP, \HH).
$$
For each $[T, N', P, H]$,
at line $19$, we generate
a pre-regular semi-algebraic system $[\BB,T,N'_{>}\cup\PP]$.
By Proposition~\ref{Proposition:borderpolynomial}, we have
$$
\begin{array}{l}
Z_{\R}(\sat{T}, \NN, \PP, \HH)=\\
Z_{\R}(B_{\neq}, T, N'_{>}\cup P_{>})
\cup Z_{\R}(\sat{T}\cup\{\Pi_{b\in B}b\}, \NN, \PP, \HH),
\end{array}
$$
which implies that
$$
\begin{array}{rcl}
Z_{\R}(\S)&=&\cup_T Z_{\R}(B_{\neq}, T, N'_{>}\cup P_{>})\\
&&\cup_T Z_{\R}(\sat{T}\cup\{\Pi_{b\in B}b\}, \NN, \PP, \HH).
\end{array}
$$
So Algorithm~\ref{Algo:GeneratePreRegularSas}
satisfies its specification.

\begin{algorithm}
\dontprintsemicolon
\linesnumbered
\caption{\SamplePoints{A, k}\label{Algo:SamplePoints}}
\KwIn{$A \subset \Q[x_1,\ldots,x_k]$ is a finite set of non-zero polynomials}
\KwOut{ A finite subset of $\Q^k$ contained in \\
       $(\Pi_{p\in A} \, p)\neq0$
        and having a non-empty intersection with 
        each connected component  of $(\Pi_{p\in A} \, p)\neq0$. }
    \uIf{$k=1$}{
         return one rational point from each 
        connected component  of  $\Pi_{p\in A} \,p \neq 0$\;
    }
    \Else{
         $A_k:=\{p\in A \mid \mvar{p}=x_k\}$; $A':=\oproj{A,x_k}$\;
         \For{$s\in\SamplePoints{A', k-1}$}{              
              Collect in a set $S$ one rational point from each 
        connected component  of $\Pi_{p\in A_k}p(s,x_k) \neq 0$;\;
              {\bf for} $\alpha\in S$ {\bf do} output $(s, \alpha)$
         }
    }
\end{algorithm}

\begin{algorithm}
\dontprintsemicolon
\linesnumbered
\caption{\LazyRealTriangularize{\S}\label{Algo:LazyRealTriangularize}}
\KwIn{a semi-algebraic system $\S=[F, \NN, \PP, \HH]$}
\KwOut{a lazy triangular decomposition of $\S$
}
    $\T := \GeneratePreRegularSas{F, N, P, H}$\;
    \For{$[B, T,  P']\in\T$}{
		 ($D, \mc{R}) = \GenerateRegularSas{B, T, P'}$\;
                 \lIf{$\mc{R}\neq\emptyset$}{ output $\mc{R}$}\;
    }
\end{algorithm}

\begin{algorithm}
\dontprintsemicolon
\linesnumbered
\caption{\RealTriangularize{\S}\label{Algo:RealTriangularize}}
\KwIn{ a semi-algebraic system $\S=[F, \NN, \PP, \HH]$}
\KwOut{a triangular decomposition of $\S$}
    $\T := \GeneratePreRegularSas{F, N, P, H}$\;
    \For{$[B, T,  P']\in\T$}{
          ($D, \mc{R}) = \GenerateRegularSas{B, T, P'}$\;
         \lIf{$\mc{R}\neq\emptyset$}{ output $\mc{R}$}\;
         \For{$p\in D$}{
         output $\RealTriangularize{F\cup\{p\}, N, P, H}$\;
         }
    }
\end{algorithm}

\smallskip\noindent{\small \bf Proof of
Algorithms~\ref{Algo:GenerateRegularSas} and~\ref{Algo:SamplePoints}}.
By the definition of ${\rm oproj}$, 
Algorithm~\ref{Algo:SamplePoints} terminates and satisfies its specification.
By Theorem~\ref{thm:dpoap}, 
${\oaf}(B)$ is an {\small FPS}.
Thus, by the definition of an {\small FPS},
Algorithm~\ref{Algo:GenerateRegularSas}
terminates and satisfies its specification.

\smallskip\noindent{\small \bf Proof of Algorithm~\ref{Algo:LazyRealTriangularize}}.
Its termination is obvious.
Let us prove the algorithm is correct.
Let $R_i$, $i=1\cdots t$ be the output. 
By the specification of each sub-algorithm, 
each $R_i$ is a regular semi-algebraic system and we have:
$$\cup_{i=1}^tZ_{\R}(R_i)\subseteq Z_{\R}(\S).$$
Next we show that 
there exists an ideal $\mc{I} \subseteq {\Q}[\x]$,
whose dimension is less than $\dim(Z(F, P_{\neq}\cup H_{\neq}))$
and 
such that 
$Z_{\R}(\S) \setminus \cup_{i=1}^t Z_{\R}(R_i) \subseteq Z_\R(\mc{I})$
holds.

At line $1$, 
by the specification of Algorithm~\ref{Algo:GeneratePreRegularSas}, 
we have
$$
\begin{array}{rcl}
Z_{\R}(\S)&=&\cup_T Z_{\R}(B_{\neq}, T,  P'_{>})\\
&&\cup_T Z_{\R}(\sat{T}\cup\{\Pi_{b\in B} \, b\}, \NN,  P_{>}, \HH).
\end{array}
$$
At line $3$, 
by the specification of Algorithm~\ref{Algo:GenerateRegularSas},
for each $B$, we compute a set $D$ such that $\factor{B}\subseteq D$
and 
$$
\cup_T Z_{\R}(D_{\neq}, T,  P'_{>}) = \cup_{i=1}^t Z_{\R}(R_i)
$$
both hold.
Combining the two relations together, we have
$$
\begin{array}{rcl}
Z_{\R}(\S)&=&\cup_T Z_{\R}(R_i)\\
&&\cup_T Z_{\R}(\sat{T}\cup\{\Pi_{p\in D}\,p\}, \NN,  P_{>}, \HH).
\end{array}
$$
Therefore, the following relations hold
$$
\begin{array}{l}
~~~Z_{\R}(\S)\setminus \cup_{i=1}^t Z_{\R}(R_i)\\
\subseteq \cup_T Z_{\R}(\sat{T}\cup\{\Pi_{p\in D}\,p\}, \NN, \PP, \HH) \\
\subseteq Z_{\R}\left( \cap_T\left(\sat{T}\cup\{\Pi_{p\in D}\,p\}\right) \right).
\end{array}
$$
Define 
$${\cal I}=\cap_T\left(\sat{T}\cup\{\Pi_{p\in D}\,p\}\right).$$
Since each $p\in D$ is regular modulo $\sat{T}$, 
we have
$$
\dim({\cal I}) < \dim\left(\cap_T \sat{T}\right)\leq \dim (Z(F, P_{\neq}\cup \HH)).
$$
So all $R_i$ form a lazy triangular decomposition of $\S$.
$\square$

\smallskip\noindent{\small \bf Proof of Algorithm~\ref{Algo:RealTriangularize}}.
For its termination, it is sufficient to prove that
there are only finitely many recursive calls to {\sf RealTriangularize}. 
Indeed, if $[F, N, P, H]$ is the input of a call to {\sf RealTriangularize}
then each of the immediate recursive calls takes $[F\cup\{p\}, N, P, H]$
as input, where $p$ belongs to the 
set $D$ of some pre-regular semi-algebraic system $[D_{\neq}, T, \PP]$.
Since $p$ is regular (and non-zero) modulo $\sat{T}$ we have:
$$\langle F \rangle \subsetneq \langle F\cup\{p\} \rangle.$$
Therefore, the algorithm terminates 
by the ascending chain condition on ideals of ${\Q}[\x]$.
The correctness of Algorithm~\ref{Algo:RealTriangularize}
follows from the specifications of the 
sub-algorithms.
$\square$

\smallskip\noindent{\small \bf Table 1}.
Table 1 summarizes the notations used in Tables 2 and 3. Tables 2 and 3
demonstrate benchmarks running in {\sc Maple} 14 ${\beta} \ 1$, 
using an Intel Core 2 Quad {\small CPU} (2.40{\small GHz}) 
with 3.0{\small GB} memory.
The timings are in seconds and 
the time-out is 1 hour.

\begin{figure}
\centering
{\textbf{Table 1} Notations for Tables 2 and 3}
\medskip
\newline
{\small
\begin{tabular}{l|l}
\hline
symbol & meaning\\\hline
\#e  & number of equations in the input system\\
\#v & number of variables in the input equations\\
d & maximum total degree of an input equation\\
G & {\sf Groebner:-Basis} ({\sf plex} order) in {\sc Maple}\\
T & {\sf Triangularize} in {\RegularChains} library of {\sc Maple}\\
LR & {\sf LazyRealTriangularize} implemented in {\sc Maple}\\
R &{\sf RealTriangularize} implemented in {\sc Maple}\\
Q & {\sc Qepcad~b}\\
$>$ $1h$& computation does not complete within 1 hour\\
FAIL & {\sc Qepcad~b} failed due to prime list exhausted\\\hline
\end{tabular}
}
\end{figure}

\smallskip\noindent{\small \bf Table 2}.
The systems in this group  involve equations only.
We report the running times for a 
triangular decomposition of the input algebraic variety
and a lazy triangular decomposition of the corresponding real variety.
These illustrate the good performance of 
our tool.

\begin{figure}
\centering
{\textbf{Table 2} Timings for varieties}
\medskip
\newline
{\small
\begin{tabular}{c|c|c|c|c}
\hline
system                  &\#v/\#e/d& G            & T     & LR\\\hline
Hairer-2-BGK            & 13/ 11/ 4 & 25         &1.924  & 2.396\\
Collins-jsc02           & 5/ 4/ 3   & 876        &0.296  & 0.820\\
Leykin-1                & 8/ 6/ 4   & 103        &3.684  & 3.924\\
8-3-config-Li           & 12/ 7/ 2  & 109        &5.440  & 6.360\\
Lichtblau               & 3/ 2/ 11  & 126        &1.548  & 11\\
Cinquin-3-3             & 4/ 3/ 4   & 64         &0.744  & 2.016\\
Cinquin-3-4             & 4/ 3/ 5   & $>$ $1h$   &10     & 22\\
DonatiTraverso-rev      & 4/ 3/ 8   & 154        &7.100  & 7.548\\
Cheaters-homotopy-1     & 7/ 3/ 7   & 3527       &174    & $>$ $1h$\\
hereman-8.8             & 8/ 6/ 6   & $>$ $1h$   &33     & 62\\
L                       &12/ 4/ 3   & $>$ $1h$    & 0.468 & 0.676\\
dgp6                    &17/19/ 2   & 27         & 60    & 63\\
dgp29                   & 5/ 4/ 15  & 84         & 0.008 & 0.016\\
\hline
\end{tabular}
}
\end{figure}

\smallskip\noindent{\small \bf Table 3}.
The examples in this table 
are quantifier elimination problems
and most of them
involve both equations and inequalities.
We provide
the timings for computing a lazy and a full
triangular decomposition of the 
corresponding semi-algebraic system
and the timings for solving the quantifier
elimination problem via {\sc Qepcad b}~\cite{Bro03} 
(in non-interactive mode).
Computations complete with our tool on 
more examples than with {\sc Qepcad b}.

\begin{figure}
\centering
{\textbf{Table 3} Timings for semi-algebraic systems}
\medskip
\newline
{\small
\begin{tabular}{c|c|c|c|c|c}
\hline
system        &\#v/\#e/d&    T    & LR      & R    & Q\\\hline
BM05-1        &4/ 2/ 3  & 0.008   & 0.208  & 0.568   &86\\
BM05-2        &4/ 2/ 4  & 0.040   & 2.284  &$>1h$  &FAIL\\
Solotareff-4b &5/ 4/ 3  & 0.640   & 2.248  &924      &$>1h$\\
Solotareff-4a &5/ 4/ 3  & 0.424   & 1.228  &8.216    &FAIL\\
putnam        &6/ 4/ 2  & 0.044   & 0.108  &0.948    &$>1h$\\
MPV89         &6/ 3/ 4  & 0.016   & 0.496  & 2.544   &$>1h$\\
IBVP          &8/ 5/ 2  & 0.272   & 0.560  &12   &$>1h$\\
Lafferriere37 &3/ 3/ 4  & 0.056   & 0.184  &0.180    &10\\
Xia           &6/ 3/ 4  & 0.164   & 191    & 739     &$>1h$\\
SEIT          &11/ 4/3  & 0.400   & $>1h$&$>1h$  &$>1h$\\
p3p-isosceles &7/ 3/ 3  & 1.348   & $>1h$&$>1h$  &$>1h$\\
p3p           &8/ 3/ 3  & 210     & $>1h$&$>1h$  &FAIL\\
Ellipse       &6/ 1/ 3  & 0.012   & $>1h$&$>1h$  &$>1h$\\
\hline
\end{tabular}
}
\end{figure}

\smallskip\noindent{\small \bf Remark}. 
The output of our tools 
is a set of regular semi-algebraic systems,
which is different than that of {\sc Qepcad b}.
We note also that our tool is more effective for
systems with more equations than inequalities.

\bigskip\noindent{\bf Acknowledgments.}
The authors would like to thank the referees for their valuable
remarks that helped to improve the presentation of the work.

\newpage

\small

\end{document}